%% file: levinwengroupcategories.tex
\documentclass[aps,prb,twocolumn,superscriptaddress,showpacs]{revtex4}
\usepackage{amsmath} 
\usepackage{amsthm} 
\usepackage{amsfonts} 
\usepackage{amssymb} 
\usepackage{mathrsfs} 
\usepackage{bbold} 
\usepackage{hyperref}
  \hypersetup{colorlinks,bookmarksopen,bookmarksnumbered,citecolor=blue,pdfstartview=FitH}
\usepackage[all]{xy}
\usepackage{graphicx}
\usepackage[usenames,dvipsnames]{color}
\usepackage{verbatim}
\usepackage{tabularx} 
\usepackage{subfigure}

\newcommand\be{\begin{equation}}
\newcommand\ee{\end{equation}}
\newcommand{\bpm}{\begin{pmatrix}}
\newcommand{\epm}{\end{pmatrix}}
\newcommand{\bmm}{\begin{matrix}}
\newcommand{\emm}{\end{matrix}}

\begin{document}

\title{Counterexamples in the Levin-Wen model, group categories, and Turaev unimodality}

\author{Spencer D. Stirling}
\email{stirling@physics.utah.edu}
\affiliation{Department of Physics and Astronomy,
  University Of Utah, Salt Lake City, UT 84112, USA}
\affiliation{Department of Mathematics,
  University of Utah, Salt Lake City, UT 84112, USA}

\date{\today}



\newcommand{\Vcat}{\mathcal{V}}
\newcommand{\id}{\text{id}}
\newcommand{\grp}{\mathcal{D}}
\newcommand{\trio}{{(\grp,q,c)}}
\newcommand{\triotrunc}{{(\grp,q)}}
\newcommand{\grpcat}{\mathscr{C}_\triotrunc}
\newcommand{\qmodz}{\mathbb{Q}/\mathbb{Z}}
\newcommand{\Hom}{\text{Hom}}
\newcommand{\modfunct}{\mathscr{F}}
\newcommand{\unitobj}{\mathbb{1}}
\newcommand{\idmat}{I}
\newcommand{\cplx}{\mathbb{C}}
\newcommand{\real}{\mathbb{R}}
\newcommand{\integers}{\mathbb{Z}}
\newcommand{\ribv}{\text{Rib}_\Vcat}
\newcommand{\ribi}{\text{Rib}_I}
\newcommand{\progplanar}{\text{ProgPlanar}}
\newcommand{\progthreed}{\text{Prog3D}}
\newcommand{\polarplanar}{\text{PolarPlanar}}
\newcommand{\piv}{\text{piv}}
\newcommand{\qtr}{\text{tr}_q}
\newcommand{\qdim}{\text{dim}_q}

\setcounter{secnumdepth}{1}
\numberwithin{equation}{section}

\theoremstyle{plain}
\newtheorem{theorem}[equation]{Theorem}
\newtheorem{proposition}[equation]{Proposition}
\newtheorem{corollary}[equation]{Corollary}
\newtheorem{lemma}[equation]{Lemma}
\newtheorem{fact}[equation]{Fact}
\newtheorem{conjecture}[equation]{Conjecture}
\theoremstyle{definition}
\newtheorem{definition}[equation]{Definition}
\newtheorem{example}[equation]{Example}
\theoremstyle{remark}
\newtheorem{remark}[equation]{Remark}

\begin{abstract}
We remark on the claim that the string-net model of Levin and Wen \cite{levin_wen}
is a microscopic
Hamiltonian formulation of the Turaev-Viro topological quantum field theory \cite{turaevviro}.
Using simple counterexamples
we indicate where interesting extra structure may be needed in the Levin-Wen model for
this to hold (however we believe that some form of the correspondence is true).

In order to be accessible to the condensed matter community we provide a very brief and
gentle introduction to the relevant concepts in category theory (relying heavily on analogy
with ordinary group representation theory).  Likewise, some physical ideas are 
briefly surveyed for
the benefit of the more mathematical reader.

The main feature of group categories under consideration is Turaev's unimodality.  We pinpoint
where unimodality should fit into the Levin-Wen construction, and show that the simplest example
\cite{levin_wen} fails to be unimodal.  
Unimodality is straightforward to compute for group categories, and we provide a complete classification
at the end of the paper.

\end{abstract}
\maketitle



\section{Introduction}
In this note we briefly consider some aspects in the
relationship between the Levin-Wen string-net model \cite{levin_wen}
and the Turaev-Viro/Barrett-Westbury state sum model \cite{turaev,barrettwestbury}.
Although some work towards a correspondence proof has
been provided \cite{rasetti1,rasetti2}, we emphasize an important extra
structure appears necessary in the Levin-Wen model.  This structure may have physical implications.

Although the details come from category theory, for the physics community we use simple
analogies with ordinary group representation theory.  Also, our discussion of the physics is kept
brief and simple so that interested mathematicians may have access.

The key issue is that $6$j-symbols with \textit{tetrahedral symmetry} play a defining role in both models.
In order to build $6$j-symbols with tetrahedral symmetry,
Turaev's construction uses the \textit{unimodality} condition.
In particular unimodality assures that the ``band-breaking'' maneuver depicted in
Fig~(\ref{fig:turaevbandbreak}) is well-defined.

However, in the Levin-Wen model the analogous maneuver
(see Fig~(\ref{fig:unimodalmove})) is implicitly allowed without restriction.  To highlight
how this may be problematic, we show here that the first example computed by Levin-Wen 
is, in fact, \textbf{not} unimodal.  This implies that the Turaev construction cannot be applied.

This example, along with most of the examples computed by Levin-Wen, are
\textit{group categories} \cite{stirling_thesis}.
At the end we formulate a theorem that clarifies conditions when a group category is unimodal.

Before proceeding into the heart of the examples, we mention the recent work of
Hong \cite{seungmoonhong} that (partially) generalizes the Levin-Wen model to unitary spherical
categories using so-called \textit{mirror conjugate symmetry} rather than tetrahedral
symmetry.  It would be interesting to explore the examples considered here in this
more general context.

\section{Levin-Wen model}
Levin-Wen \cite{levin_wen} consider a model on a fixed trivalent graph (with
oriented edges) embedded in $2$d.  A
typical configuration is pictured in Fig~(\ref{fig:samplestringnet}) where
each edge is labelled by a \textbf{string type} $j$.  There are finitely-many string
types, hence we simply refer to them by number $\{0,1,\ldots,N\}$.  

There is a duality $j\mapsto j^*$
that satisfies $()^{**}=()$, and the $0$ label means ``no string'', hence we require
$0^*=0$.  We are allowed to reverse the orientation
of an edge if we reverse its label $j\mapsto j^*$.
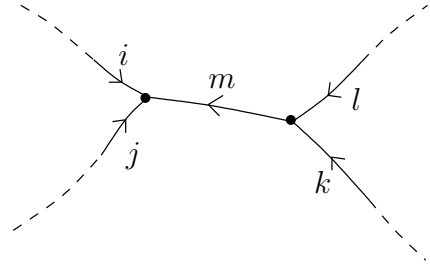
\begin{figure}[!htb]
  \centering
  \scalebox{1}{\input{samplestringnet.pspdftex}}
  \caption{A sample string-net configuration: a trivalent lattice in $2$-dimensions 
   with strings ``colored'' by integer string types $i$.  In category language the labels are
   simple objects $V_i$.  The dashed lines indicate that
   the string-net continues outside of the picture.}
  \label{fig:samplestringnet}
\end{figure}

In order to quantize this Levin and Wen introduce a Hilbert space.  It is defined by
promoting every configuration $X$ to be an orthonormal basis vector $\vert X\rangle$.  An arbitrary state
in the Hilbert space is a formal linear combinations of these.
Given any string-net configuration $\vert X\rangle$, 
an arbitrary state $\vert\Phi\rangle$ has an associated probability amplitude
$\Phi(X):=\langle X|\Phi \rangle$ of being in that particular configuration.
From now on we do not differentiate between states $\vert \Phi\rangle$ and their wavefunctions $\Phi$.

We are interested in a subspace of states $\{\Phi\}$ that encodes the topological information of the phase.  
One of their main results \cite{levin_wen} is that $\{\Phi\}$ can be realized as the ground state
subspace of an exactly-soluble Hamiltonian.

Rather than write down an explicit Hamiltonian, it is instructive to review the
strategy that Levin-Wen use to find the topological subspace $\{\Phi\}$.  The idea
is based on \textit{renormalization semigroup (RG)} flow.  For concreteness suppose that
we start with a (possibly too complicated) Hamiltonian $H$.

Usually $H$ has some parameters that can be adjusted (e.g. coupling constants), 
and as we adjust these parameters
we also deform its ground state(s) $\Phi$ (see Fig~(\ref{fig:rgflow})).  
Hence we actually have a family $\{(H_\alpha,\Phi_\alpha)\}_\alpha$ of
different (but related) theories where $\alpha$ denotes all of the adjustable parameters.
\begin{figure}[!htb]
  \centering
  \scalebox{1}{\input{rgflow.pspdftex}}
  \caption{Cartoon phase diagram}
  \label{fig:rgflow}
\end{figure}
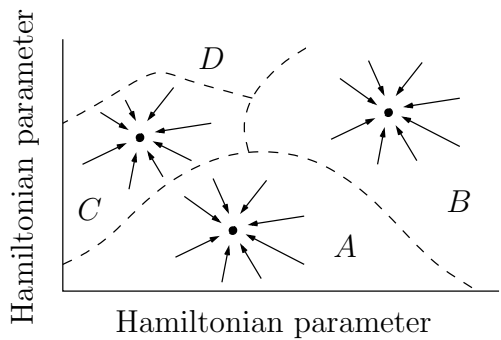

On the other hand, often we do not know the exact Hamiltonian $H$ (nor the true parameters $\alpha$), or if we
do then its complexity is intractable.  Fortunately, for a fixed set of parameters $\alpha$
(say in the ``$A$'' region in Fig~(\ref{fig:rgflow}))
the Hamiltonian $H_\alpha$ is usually
approximated by a much simpler \textit{effective} Hamiltonian $H_{A(\alpha)}$
(equipped with its own adjustable parameters $A(\alpha)$ that are in principle 
derived from the original parameters $\alpha$).

If we adjust the original parameters $\alpha$ for $H_\alpha$ slightly, then 
the effective parameters $A(\alpha)$ adjust slightly. 

However, if we make a large adjustment to the parameters $\alpha$ (e.g. $\alpha$ is adjusted
into the ``$B$'' region) then
we cannot expect that the same approximations remain valid.  
Instead a different set of approximations may be appropriate, yielding a
different form for the effective Hamiltonian $H_{B(\alpha)}$.

In this way the family of theories $\{(H_\alpha,\Phi_\alpha)\}_\alpha$ is carved into
\textbf{phases}.  A phase $A$ is a region of parameters $\alpha$ where the Hamiltonians can all
can be approximated similarly.  Although the effective Hamiltonians still depend on adjustable
parameters $A(\alpha)$, the Hamiltonians are of the same \textit{form}.

Assume that we pick a Hamiltonian $H_{A}$ in the $A$ phase (we no longer refer
to the original parameters $\alpha$).
The RG flow procedure involves performing an iterated \textit{scaling out} while averaging the
details.  From experience we expect that this preserves the form of the Hamiltonian, but may
adjust the coupling constants, etc.  In other words, scaling is just a particular recipe to
flow to a new Hamiltonian $H_{A^\prime}$ that is in the same phase.  The new Hamiltonian is
simpler because at each scaling iteration we average out the details.

Eventually scaling produces a system that is minimally simplistic, and further scaling does
not lead to further simplification.  This produces a scale-invariant fixed point theory 
$(H,\Phi)_{\text{fixed}}$, 
i.e. a conformal field theory (CFT).

Since $\Phi_{\text{fixed}}$ is scale independent the
long-range properties of any other groundstate $\Phi_A$ in the phase 
are well-approximated by $\Phi_{\text{fixed}}$.
Hence in the following we concentrate only on $\Phi_{\text{fixed}}$,
and we drop the ``fixed'' subscript.

\subsection{6j-symbols}

Since the scale-invariant fixed point is somehow the ``simplest'' system, we
have the best chance of identifying the topological subspace there.
On the other hand, Levin-Wen make the ansatz that the topological phase is completely determined
by the data from a \textbf{modular tensor category} \cite{turaev}.

Combining these ideas,
they propose a set of rules at the fixed point that determines how any state $\Phi$ in the topological subspace
transforms under graphical deformations.

Before describing the rules, it is useful to discuss the origins of the basic data in category
theory.  Actually, they really 
only need some of the data, starting
with the list of string types $\{0,\ldots,N\}$ described above.  In category language the string types are
\textbf{simple objects} $\{V_0,V_1,\ldots,V_N\}$, and they can be viewed as generalized irreducible
group representations.  They satisfy fusion rules (generalized Clebsch-Gordan rules) 
\begin{equation}
\label{eq:fusionrules}
  V_i\otimes V_j = \oplus N^k_{ij}\,V_k\quad N^k_{ij}\in\integers 
\end{equation}
Order is important in the generalized tensor product that we encounter in
category theory (i.e. $V_i\otimes V_j\neq V_j\otimes V_i$).  From now on
we denote the integers $\{0,\ldots,N\}$ by $I$.

For simplicity, Levin and Wen restrict themselves to fusion rules satisfying
\begin{equation}
  N^k_{ij}=0\text{ or }1
\end{equation}
in which case the fusion rules are called \textbf{branching rules}.

The dual representation $(V_i)^*$ of an irrep is also an irrep, 
hence we denote it by $V_{i^*}$ with the
convention that $i^{**}=i$.  We shall see that this harmless-looking notation is at the
core of the problem, i.e. in a category this rule is subject to some conditions that
are automatically satisfied in ordinary group representation theory. 

Given the simple objects and branching rules $\delta_{ijm}:=N^{m^*}_{ij}=0\text{ or }1$, 
the necessary data is provided by
a system of tensors 
\begin{equation}
  (F^{ijm}_{kln},d_i)
\end{equation}
satisfying some consistency equations (there may be many solutions)
\begin{align}
  & F^{ijm}_{j^* i^* 0}=\frac{\sqrt{d_m}}{\sqrt{d_i}\sqrt{d_j}}\delta_{ijm}\\
  & F^{ijm}_{kln}=F^{lkm^*}_{jin}=F^{jim}_{lkn^*}=
    F^{imj}_{k^*nl}\frac{\sqrt{d_m}\sqrt{d_n}}{\sqrt{d_j}\sqrt{d_l}}\nonumber\\
  & \sum_{n\in I}F^{mlq}_{kp^*n}F^{jip}_{mns^*}F^{js^*n}_{lkr^*}=F^{jip}_{q^*kr^*}F^{riq^*}_{mls^*}\nonumber
\end{align}
The first line is a normalization condition.  The second is the tetrahedral symmetry, and the
third is the Biedenharn-Elliott identity.

If we know the underlying modular tensor category then we can find these tensors explicitly.  
$F^{ijm}_{kln}\in\cplx$ is the \textbf{quantum $6j$-symbol} and 
each $d_i\in\cplx$ is the \textbf{quantum dimension}
of $V_i$.
\footnote{More appropriately these should be called ``braided $6j$-symbol'' and 
``braided dimensions'' since ordinary group representation theory is also ``quantum'', i.e.
bosons and fermions.}
These notions are generalizations of the ordinary $6j$-symbols and vector space dimensions 
associated to irreps.

Returning to a wavefunction $\Phi$ in the topological subspace, at the fixed point
Levin-Wen propose transformation rules as in Fig~(\ref{fig:wavefunctionansatz})
(this list is not complete).
Since any graph
can be popped down to nothing using these rules, $\Phi$ is determined uniquely (however if we
embed the graph in a $2$d surface with nontrivial topology then $\Phi$ is not unique). 
\begin{figure}[!htb]
  \centering
  \input{wavefunctionansatz.pspdftex}
  \caption{Conditions to determine fixed-point ground state $\Phi$}
  \label{fig:wavefunctionansatz}
\end{figure}

\section{Group categories}
Group categories (or \textit{pointed} categories) have been studied by a variety of
authors in a variety of contexts (e.g. \cite{quinn} \cite{frolich_kerler} \cite{joyal_street}).
We recommend \cite{stirling_thesis} for further
details and compatible notation.

\subsection{Braiding}

Ordinary group representation theory is a reasonable prototype to model modular tensor
categories.  
\footnote{Actually we neither require nor desire the full structure of modular tensor categories,
hence instead we can consider the simpler structure of \textbf{semisimple ribbon Ab-categories}.
It happens that ordinary group representation theory is a closer prototype for semisimple ribbon
categories anyway.}
Consider a many-body bosonic or fermionic system in $(2+1)$-dimensions.  
Each elementary particle is a copy of an irreducible representation $V_i$.
If we have two particles then we have the tensor product (again ordering is
important)
\begin{equation}
  V_i\otimes V_j
\end{equation}
(which, of course, can be decomposed using the fusion rules in equation~(\ref{eq:fusionrules})).
For bosons we have an exchange rule
\begin{align}
  \text{Perm}:&V_i\otimes V_j\xrightarrow{\sim} V_j\otimes V_i\\
              &v\otimes w\mapsto w\otimes v\quad\quad v\in V_i, w\in V_j\nonumber
\end{align}
and for fermions we have a different exchange rule
\begin{align}
  \text{Perm}:&V_i\otimes V_j\xrightarrow{\sim} V_j\otimes V_i\\
              &v\otimes w\mapsto -w\otimes v\quad\quad v\in V_i, w\in V_j\nonumber
\end{align}
In both cases we have
\begin{equation}
  \text{Perm}^2=\idmat\quad\quad\text{(the identity matrix)}
\end{equation}
In this way both bosonic and fermionic systems are considered ``commutative'' since we know what
happens when we exchange the particles $V_i\otimes V_j\xrightarrow{\sim}V_j\otimes V_i$ (and two permutations
is equivalent to doing nothing).
Confusingly, in category language they are both called \textbf{symmetric}
theories (despite the usual physical nomenclature ``antisymmetric'' for fermions).

Categories allow a much more interesting \textit{weakened} form of commutativity: \textbf{braiding}.
Here we have invertible braiding matrices
\begin{equation}
  c_{V_i,V_j}:V_i\otimes V_j\xrightarrow{\sim} V_j\otimes V_i
\end{equation}
that do \textit{not} satisfy $c_{V_j,V_i}\circ c_{V_i,V_j}=\idmat$, but rather a more
elaborate set of conditions - the \textbf{hexagon relations}.

For details concerning the hexagon relations
we refer the reader to standard references
\cite{turaev}, \cite{bakalov_kirillov}, \cite{kassel}, \cite{joyal_street}.  It turns out that the hexagon
relations come from an obvious picture.

\subsection{Braiding for group categories}
The hexagon relations for group categories are exhaustively solved \cite{stirling_thesis}.
In fact the structure is much richer than will be apparent here.

Every group category $\grpcat$ is constructed from some fundamental data $\triotrunc$.  $\grp$ is
any finite abelian group, and $q$ is a function on $\grp$ that returns
a complex phase $\exp(2\pi i q)$.  Since $q$ is a phase it is only well-defined up to
the interval $[0,1]$.  In fact $q$ takes values in $\qmodz$.  The function $q$ must also be a
\textit{pure quadratic form}.

To understand the quadratic form the reader should think about quadratic functions on real numbers
$q(x):=\frac{1}{2}x^2$.  We can easily define an induced bilinear function:
$b(x,y):=q(x+y)-q(x)-q(y)=\frac{1}{2}(x+y)^2-\frac{1}{2}(x)^2-\frac{1}{2}(y)^2=xy$.

Likewise, a pure quadratic form is a function such that the induced function
\begin{align}
  &b:\grp\otimes\grp\rightarrow\qmodz\\
  &b(x,y) := q(x+y)-q(x)-q(y) \pmod{1}\nonumber
\end{align}
is bilinear.  The adjective ``pure'' means
\begin{equation}
\label{eq:purity}
  q(nx)\equiv n^2 q(x) \pmod{1}\quad\quad n\in\integers
\end{equation}

\begin{example}
\label{ex:z2groupcategory}
The simplest example is when $\grp=\integers_2=\{\hat{0},\hat{1}\}_+$, i.e. 
$\hat{1}+\hat{1}=\hat{2}\equiv \hat{0} \pmod{2}$ (we use the hat to differentiate
elements of the group from numbers in $\qmodz$).
Then, since 
\begin{equation}
  2 b(\hat{1},\hat{1})\equiv b(2\hat{1},\hat{1})\equiv b(\hat{2},\hat{1})
  \equiv b(\hat{0},\hat{1})\equiv 0 \pmod{1}
\end{equation}
(bilinearity is used in every step, and $b$ is valued in $\qmodz$) we have
two possibilities for $b$:
\begin{equation}
  b(\hat{1},\hat{1}) \equiv 0 \pmod{1}\quad\text{ or }\quad b(\hat{1},\hat{1}) \equiv \frac{1}{2} \pmod{1}
\end{equation}
Because of bilinearity the function $b$ is fully determined by its values on the generator $\hat{1}$
of the group $\grp$.  The same statement also holds for the quadratic form $q$ by using the purity
condition in equation~(\ref{eq:purity})).

Suppose we consider the first case, i.e. $b(\hat{1},\hat{1}) \equiv 0 \pmod{1}$.  Then there are
two possible pure quadratic forms that produce this $b$:
\begin{equation}
  q(\hat{1})\equiv 0 \pmod{1} \quad\text{ or }\quad q(\hat{1})\equiv \frac{1}{2} \pmod{1}
\end{equation}
It turns out that \textit{both} group categories $\grpcat$ (defined below)
constructed from these two quadratic forms
produce the same data $(F^{ijm}_{kln},d_i)$ and hence are the same (from the Levin-Wen
perspective).  They both correspond to $\integers_2$-lattice gauge theory.

More interesting examples occur when $b(\hat{1},\hat{1}) \equiv \frac{1}{2} \pmod{1}$.
Then there are
two possible pure quadratic forms that produce this $b$:
\begin{equation}
  q(\hat{1})\equiv \frac{1}{4} \pmod{1} \quad\text{ or }\quad q(\hat{1})\equiv \frac{3}{4} \pmod{1}
\end{equation}
The group category $\grpcat$ from the LHS corresponds \cite{stirling_thesis}
to $U(1)$ Chern-Simons at level $2$.  The other one corresponds to $U(1)$
Chern-Simons at level $-2$.

On the other hand, both theories produce the same data $(F^{ijm}_{kln},d_i)$ and
hence (from the Levin-Wen perspective) both produce the same \textit{doubled} Chern-Simons theory.
We shall show below that \textbf{neither example is unimodal}.

Since this is the first example considered
by Levin and Wen\cite{levin_wen},
we already encounter the example promised in the introduction.
\end{example}

We continue our construction of the group category $\grpcat$ given the data
$\triotrunc$.
The simple objects (string types) and fusion rules (branching rules) are easy to define.
The list of string types $\{0,\ldots,N\}$ is replaced by $\grp$.
Hence we use string labels like $i,j,k\in I$ interchangeably
with group elements $x,y,z\in\grp$ often here.

For every $x\in\grp$ we
define a simple object (these were denoted $V_i$ above)
\begin{equation}
  \cplx_x\quad x\in\grp
\end{equation}
which is a 1-dimensional complex vector space \textit{graded} by the group element $x$.

The fusion rules are also easy to define, and indeed satisfy the 
\textit{branching rule} condition $N^z_{xy}=0\text{ or }1$:
\begin{equation}
  \cplx_x\otimes \cplx_y = \delta_{z,x+y}\cplx_{z}
\end{equation}
Here $\delta$ is the Kronecker delta.

The braiding is slightly more subtle.  First, given the simple fusion rules, it is clear that the
braiding matrix
\begin{align}
  c_{x,y}:&\cplx_x\otimes\cplx_y \xrightarrow{\sim} \cplx_y\otimes\cplx_x\\
          &\cplx_{x+y}\xrightarrow{\sim}\cplx_{x+y}\nonumber
\end{align}
must be multiplication by a complex number.  For group categories this complex
number is a phase:
\begin{align}
  c_{x,y}:&\cplx_{x+y}\xrightarrow{\sim}\cplx_{x+y}\\
          &v_{x+y}\mapsto \exp(2\pi i s(x,y)) v_{x+y}\nonumber
\end{align}
where $v_{x+y}\in\cplx_{x+y}$ and $s(x,y)\in\qmodz$.

It remains to specify the phase angle $s(x,y)\in\qmodz$.
Every finite abelian group $\grp$ can be decomposed (non-uniquely) as a direct product
of cyclic groups $\integers_{n_1}\times\integers_{n_2}\times\ldots$.  
Pick a generator $\hat{1}_s$ for each cyclic group, and denote the order of that
cyclic group by $n_s$.  
Then an arbitrary element $x\in\grp$ can be written uniquely (once generators
are picked) as
\begin{equation}
  x=\sum_s x_s \hat{1}_s\quad\quad 0\leq x_s \leq n_s 
\end{equation}

Pick an arbitrary ordering $\hat{1}_1<\hat{1}_2<\ldots$ on the various generators.
Define for convenience
\begin{align}
  &q_s := q(\hat{1}_s)\\
  &b_{st} := b(\hat{1}_s,\hat{1}_t)\nonumber
\end{align}
(we refer the reader to \cite{stirling_thesis} for discussion concerning how these results change
for different choices of generators and orderings). 
Then if we write arbitrary elements $x,y\in\grp$ in terms of the generators
\begin{align}
  &x=\sum_s x_s \hat{1}_s\quad\quad 0\leq x_s \leq n_s\\
  &y=\sum_s y_s \hat{1}_s\quad\quad 0\leq y_s \leq n_s\nonumber  
\end{align}
then we define the braiding phase angle by
\begin{equation}
\label{eq:braidingphaseangle}
  s(x,y) := \sum_{s<t}x_s y_t b_{st} + \sum_s x_s y_s q_s
\end{equation}

\begin{example}
We revisit example~(\ref{ex:z2groupcategory}) in the case
\begin{equation}
  b(\hat{1},\hat{1})\equiv\frac{1}{2}\pmod{1}\quad\quad q(\hat{1})\equiv\frac{1}{4}\pmod{1}
\end{equation}
Then it is easy to compute the braiding matrices (here just phases)
\begin{align}
  &c_{\hat{0},\hat{0}}=\exp\left(2\pi i \left(0\cdot 0 \cdot \frac{1}{4}\right)\right)=1\label{eq:z2braidings}\\
  &c_{\hat{0},\hat{1}}=\exp\left(2\pi i \left(0\cdot 1 \cdot \frac{1}{4}\right)\right)=1\notag\\
  &c_{\hat{1},\hat{0}}=\exp\left(2\pi i \left(1\cdot 0 \cdot \frac{1}{4}\right)\right)=1\notag\\
  &c_{\hat{1},\hat{1}}=\exp\left(2\pi i \left(1\cdot 1 \cdot \frac{1}{4}\right)\right)=i\notag
\end{align}
\end{example}

\subsection{Twists}
For bosons and fermions in $(3+1)$-dimensions we have the \textit{spin-statistics theorem}.
This relates the effect of exchanging two particles (the exchange statistics when the $\text{Perm}$
operation is applied) to their individual spins.

On the other hand, the spin for an individual particle determines how it is affected under
$3$-dimensional rotations.  
Hence we may view the spin-statistics
theorem as a relationship between multi-particle $\text{Perm}$ operations and
single-particle rotation operations.

In $(2+1)$-dimensions we might imagine a similar story, except here the analogue of ``spin''
must describe what happens to an irrep (simple object) under $\text{SO}(2)$ rotations.  This is the
\textbf{twist} matrix
\begin{equation}
  \theta_i:V_i\xrightarrow{\sim}V_i
\end{equation}

Furthermore, we have seen that in $(2+1)$-dimensions $\text{Perm}$ can be generalized to
braiding $c_{V_i,V_j}$.
In categories with \textit{both} braiding and twist, it is reasonable to assume that a compatibility
relationship exists analogous to the spin-statistics theorem.  

Indeed there
is a necessary compatibility between braiding and twists 
(with an easy geometric interpretation in terms of ribbons, see e.g. \cite{stirling_wu_bcqm}) 
called \textbf{balancing}.
Balancing gives certain restrictions on twists (much like fermions have half-integer spin).

In the spirit of this paper we merely give a formula for the twist matrices for group categories.
Given a simple object $\cplx_x$ where $x\in\grp$
the twist matrix is just multiplication by a phase
\begin{align}
  \theta_x:&\cplx_x\xrightarrow{\sim}\cplx_x\\
           &v_x\mapsto \exp(2\pi i q(x))v_x\quad\quad v_x\in\cplx_x\nonumber
\end{align}
 
\begin{example}
We revisit the same example~(\ref{ex:z2groupcategory}) in the case
\begin{equation}
  b(\hat{1},\hat{1})\equiv\frac{1}{2}\pmod{1}\quad\quad q(\hat{1})\equiv\frac{1}{4}\pmod{1}
\end{equation}
Then it is easy to compute the twist matrices (here just phases)
\begin{align}
  &\theta_{\hat{0}}=\exp\left(2\pi i (0)\right)=1\label{eq:z2twists}\\
  &\theta_{\hat{1}}=\exp\left(2\pi i \left(\frac{1}{4}\right)\right)=i\notag
\end{align}
\end{example}

\subsection{Duality and ribbon categories}
Duality is a fundamental property of
representation theory.  Given an irreducible representation $V_i$ we can
always consider the dual representation $(V_i)^*$, which we dangerously
denoted $V_{i^*}$ above.

In category theory \textbf{rigidity} is the appropriate generalization, however the details can be
rather different than those in ordinary group representation theory.  For now, each simple object (string type) 
$V_i$ has another simple object associated to it:  its \textit{right} dual $(V_i)^*$.
Furthermore, we have a categorical notion of pair creation and annihilation.  These
are the \textbf{birth} and \textbf{death} matrices
\begin{align}
  &b_{V_i}:\unitobj\rightarrow V_i\otimes (V_i)^*\\
  &d_{V_i}:(V_i)^*\otimes V_i\rightarrow \unitobj\nonumber
\end{align}
($\unitobj$ is interpreted as the vacuum).

Notice that the ordering of the objects is exactly opposite for the birth and death matrices, 
hence we \textit{cannot} simply birth a pair and then annihilate it without performing
some intermediate moves (such as braiding and twisting).
\footnote{This is how the \textbf{quantum dimension} $d_i$ of a simple object $V_i$ 
is computed, for example.} 

In a category with braiding, twists, \textit{and} rigidity we may desire some compatibility
between all three structures (we already mentioned the ``balancing'' condition between
braiding and twists for example).  
Again in the spirit of this paper, rather than discuss this point further we
mention that a \textbf{ribbon category} is a category with braiding, twists, and
rigidity such that all three structures interact appropriately.  

Even for group categories rigidity is slightly subtle (again see \cite{stirling_thesis}).
However, for this paper it suffices to merely define when objects are dual to
each other.  The reader can guess that each simple object $\cplx_x$ has a right dual
\begin{equation}
  \cplx_{-x}
\end{equation}

\section{Unimodality}

We have specified enough information about $\grpcat$
in order to decide when a given group category is unimodal 
(we provide a classification theorem below).
Before doing this, however, let us motivate why unimodality is important.

In the Levin-Wen model each \textit{directed} edge is labelled 
by a string type $i$
(simple object $V_i$).  Because of the dangerous identification $(V_i)^*=V_{i^*}$
we can always
use simplified labels such as $i$ or $j$ or $m^*$, and we may always perform
a ``string-breaking'' move such as in figure~(\ref{fig:unimodalmove}).
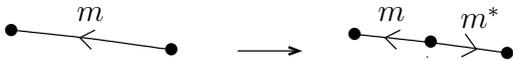
\begin{figure}[!htb]
\begin{center}
  \centering
  \input{unimodalmove.pspdftex}
  \caption{String-breaking maneuver is implicit in Levin-Wen model, but non-trivial in
    Turaev's construction.}
  \label{fig:unimodalmove}
\end{center}
\end{figure}

On the other hand, in Turaev's formulation \cite{turaev} there are non-trivial invertible
matrices
\begin{equation}
  w_{i^*}:V_{i^*}\xrightarrow{\sim}(V_i)^*
\end{equation}
that form part of the defining structure of the category.  They
also provide a ``band-breaking'' maneuver as in figure~(\ref{fig:turaevbandbreak}), however
the procedure may not be self-consistent.  It turns out that an inconsistency
may arise when simple objects are self-dual (i.e. when $V_i=V_{i^*}$),
and hence we must enforce the extra unimodality condition.
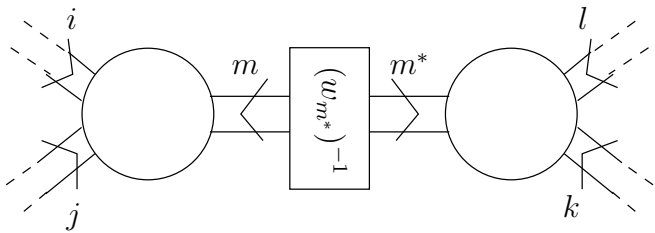
\begin{figure}[!htb]
  \centering
  \input{turaevbandbreak.pspdftex}
  \caption{Band-breaking maneuver in Turaev's construction.}
  \label{fig:turaevbandbreak}
\end{figure}

Rather than go into details, we give the unimodality condition for group categories:
suppose $\cplx_x$ is self-dual
\begin{equation}
\label{eq:unimodalcondition}
  \cplx_x = \cplx_{-x}
\end{equation}
(in other words $x=-x\in\grp$).  Then unimodality is a condition on the twist and
braiding:
\begin{equation}
  \theta_x\cdot c_{x,x}\overset{?}{=}1\quad\quad\text{whenever }x=-x
\end{equation}
If this is true for all $x=-x$ then the category is \textbf{unimodal}.

\subsection{Counterexample}

Again we revisit the case $\grp=\integers_2$ in example~(\ref{ex:z2groupcategory})
(however \textit{not} $\integers_2$-lattice gauge theory).  Given the braiding and twist
tables provided in equations~(\ref{eq:z2braidings}) and (\ref{eq:z2twists}) we compute
that for the self-dual object $\hat{1}$
\begin{equation}
  \theta_{\hat{1}}\cdot c_{\hat{1},\hat{1}}=i\cdot i=-1
\end{equation}
This example is \textbf{not unimodal} and Turaev's machinery does not apply.
However, as already discussed this is $U(1)$ Chern-Simons at level $2$ and is the
first example considered by Levin-Wen \cite{levin_wen}.

Alternatively, we can suspect a discrepancy by examining the quantum dimension
$d_{\hat{1}}$ computed in \cite{levin_wen}.  There they assert that
\begin{equation}
  d_{\hat{1}}=-1
\end{equation}
On the other hand, it is proven in $\cite{stirling_thesis}$ that for group categories
$d_x=1$ for every simple object $\cplx_x$.
\footnote{Our result also agrees with the quantum dimension computed in the quantum group
$\text{SU}_q(2)$ at level $1$.  This is relevant since
rank-level duality asserts that $\text{SU}_q(n)$ at level $k$ is isomorphic to
$\text{SU}_q(k)$ at level $n$ whenever $n,k>1$.  If $n=1$ (or $k=1)$ then rank-level duality is
slightly modified:
$\text{U}(1)$ Chern-Simons at level $k$ is isomorphic to $\text{SU}_q(k)$ at level $1$.}

\section{Unimodality and group categories}

It is straightforward to classify conditions when a group category is
unimodal:

\begin{theorem}
Let $\grpcat$ be a group category.  Arbitrarily decompose the finite abelian group $\grp$ into
cyclic groups of orders $n_s$ and generators $\hat{1}_s$ ($s$ indexes the cyclic factors).
Pick an arbitrary ordering $\hat{1}_1 < \hat{1}_2 < \ldots$ for the generators.
Then we have the following cases:
\begin{enumerate}
 \item If $\grp$ contains a cyclic factor of odd order (i.e. if one of the $n_s$ is odd) then
    $\grpcat$ is unimodal.
 \item If $\grp$ contains no odd-order cyclic factors then $\grpcat$ is unimodal if and only if
\begin{equation}
  \sum_s \frac{1}{2}(n_s)^2 q_s \in\integers
\end{equation}
\end{enumerate}
\end{theorem}
\begin{proof}
Suppose $x=-x$.  Decompose $x$ and $-x$ using the generators
\begin{align}
  &x=\sum_s x_s \hat{1}_s\quad\quad 0\leq x_s\leq n_s\\
  &-x=\sum_s (n_s-x_s) \hat{1}_s\nonumber
\end{align}
The condition that $x=-x$ implies
\begin{equation}
  n_s-x_s = x_s\quad\forall s
\end{equation}
which implies that $n_s=2x_s$.  In particular this implies that every $n_s$ must be even,
hence if there exists any $n_s$ odd then $x\neq -x$ for every $x\in\grp$. 

Now suppose that every $n_s$ is even.  Then the braiding and twist formulas imply
\begin{equation}
  \theta_x\cdot c_{x,x}=\exp(2\pi i q(x))\exp(2\pi i s(x,x))
\end{equation}
Equation~(\ref{eq:braidingphaseangle}) says
\begin{equation}
   s(x,x) := \sum_{s<t}x_s x_t b_{st} + \sum_s (x_s)^2 q_s
\end{equation}
The quadratic form $q(x)=q\left(\sum_s x_s \hat{1}_s\right)$ can be decomposed using successive
iterations of the defining formula
\begin{equation}
  q(x+y)-q(x)-q(y)=b(x,y)
\end{equation}
rearranged as $q(x+y)=b(x,y)+q(x)+q(y)$.  
For example the first iteration is
\begin{align}
  q\left(\sum_s x_s \hat{1}_s\right) &=q\left(x_1 \hat{1}_1 +\sum_{s>1} x_s \hat{1}_s\right)\\
   &= b\left(x_1 \hat{1}_1,\sum_{s>1} x_s \hat{1}_s\right)+q(x_1 \hat{1}_1)\\ 
   &\quad +q\left(\sum_{s>1} x_s \hat{1}_s\right)\nonumber\\
   &=\sum_{s>1} x_1 x_s b_{1s} + (x_1)^2 q_1 +q\left(\sum_{s>1} x_s \hat{1}_s\right)\nonumber
\end{align}
We used bilinearity of $b$ and the fact that $q$ is pure in the last equality.  
Now, concentrating on the last term, we repeat the process.  Iterating we finally obtain
\begin{equation}
  q(x)=\sum_{s<t}x_s x_t b_{st}+\sum_s (x_s)^2 q_s
\end{equation}
which is precisely the formula for $s(x,x)$.  

Hence
\begin{equation}
  \theta_x\cdot c_{x,x}=\exp(2\pi i (q(x)+s(x,x))=\exp(2\pi i 2q(x))
\end{equation}
Recall that we are only considering $x\in\grp$ such that $x=-x$.  Also recall from
that beginning of the proof that then $n_s=2 x_s$.  Hence we substitute
$x_s=\frac{1}{2}n_s$ into the expression for $2q(x)$ and obtain
\begin{align}
 2q(x)&=\sum_{s<t}2 \frac{n_s}{2} x_t b_{st}+\sum_s 2(\frac{n_s}{2})^2 q_s\\
      &=\sum_{s<t} x_t n_s b_{st}+\sum_s \frac{1}{2}(n_s)^2 q_s\nonumber
\end{align}
But $x_t n_s b_{st}=x_t b(n_s \hat{1}_s,\hat{1}_t)=x_t b(0,\hat{1}_t)\equiv 0 \pmod{1}$, hence the
first sum is always an integer and does not play a role in the exponent.

Viewing the second factor we arrive at the desired result, i.e.
\begin{equation}
  \theta_x\cdot c_{x,x}=\exp(2\pi i 2q(x))=1\Leftrightarrow \sum_s \frac{1}{2}(n_s)^2 q_s\in\integers
\end{equation}
\end{proof}

\section{Acknowledgements}
During this work SDS was supported by an FQXi grant under Yong-Shi Wu at the University
of Utah.

\bibliographystyle{apsrev}
\bibliography{levinwengroupcategories}

\end{document}

%% file: samplestringnet.pspdftex
\begin{picture}(0,0)%
\includegraphics{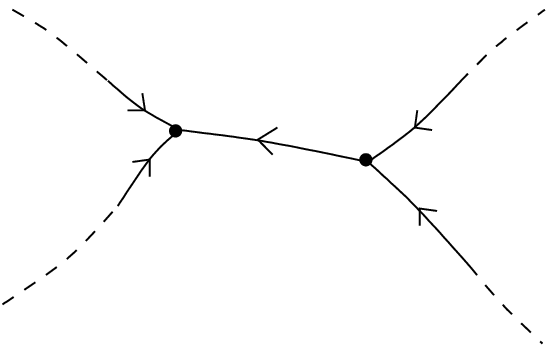}%
\end{picture}%
\setlength{\unitlength}{4144sp}%
\begingroup\makeatletter\ifx\SetFigFont\undefined%
\gdef\SetFigFont#1#2#3#4#5{%
  \reset@font\fontsize{#1}{#2pt}%
  \fontfamily{#3}\fontseries{#4}\fontshape{#5}%
  \selectfont}%
\fi\endgroup%
\begin{picture}(2503,1550)(574,-1419)
\put(1216,-241){\makebox(0,0)[lb]{\smash{{\SetFigFont{12}{14.4}{\rmdefault}{\mddefault}{\updefault}{\color[rgb]{0,0,0}$i$}%
}}}}
\put(1756,-376){\makebox(0,0)[lb]{\smash{{\SetFigFont{12}{14.4}{\rmdefault}{\mddefault}{\updefault}{\color[rgb]{0,0,0}$m$}%
}}}}
\put(1261,-826){\makebox(0,0)[lb]{\smash{{\SetFigFont{12}{14.4}{\rmdefault}{\mddefault}{\updefault}{\color[rgb]{0,0,0}$j$}%
}}}}
\put(2386,-1006){\makebox(0,0)[lb]{\smash{{\SetFigFont{12}{14.4}{\rmdefault}{\mddefault}{\updefault}{\color[rgb]{0,0,0}$k$}%
}}}}
\put(2611,-511){\makebox(0,0)[lb]{\smash{{\SetFigFont{12}{14.4}{\rmdefault}{\mddefault}{\updefault}{\color[rgb]{0,0,0}$l$}%
}}}}
\end{picture}%

%% file: rgflow.pspdftex
\begin{picture}(0,0)%
\includegraphics{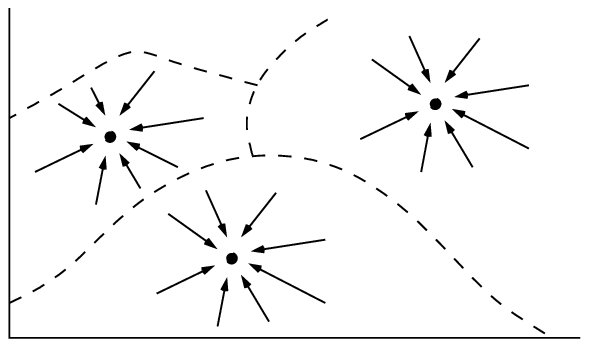}%
\end{picture}%
\setlength{\unitlength}{4144sp}%
\begingroup\makeatletter\ifx\SetFigFont\undefined%
\gdef\SetFigFont#1#2#3#4#5{%
  \reset@font\fontsize{#1}{#2pt}%
  \fontfamily{#3}\fontseries{#4}\fontshape{#5}%
  \selectfont}%
\fi\endgroup%
\begin{picture}(2962,2037)(201,-1525)
\put(2161,-1006){\makebox(0,0)[lb]{\smash{{\SetFigFont{12}{14.4}{\rmdefault}{\mddefault}{\updefault}{\color[rgb]{0,0,0}$A$}%
}}}}
\put(2836,-736){\makebox(0,0)[lb]{\smash{{\SetFigFont{12}{14.4}{\rmdefault}{\mddefault}{\updefault}{\color[rgb]{0,0,0}$B$}%
}}}}
\put(631,-781){\makebox(0,0)[lb]{\smash{{\SetFigFont{12}{14.4}{\rmdefault}{\mddefault}{\updefault}{\color[rgb]{0,0,0}$C$}%
}}}}
\put(1351,119){\makebox(0,0)[lb]{\smash{{\SetFigFont{12}{14.4}{\rmdefault}{\mddefault}{\updefault}{\color[rgb]{0,0,0}$D$}%
}}}}
\put(856,-1456){\makebox(0,0)[lb]{\smash{{\SetFigFont{12}{14.4}{\rmdefault}{\mddefault}{\updefault}{\color[rgb]{0,0,0}Hamiltonian parameter}%
}}}}
\put(361,-1411){\rotatebox{90.0}{\makebox(0,0)[lb]{\smash{{\SetFigFont{12}{14.4}{\rmdefault}{\mddefault}{\updefault}{\color[rgb]{0,0,0}Hamiltonian parameter}%
}}}}}
\end{picture}%

%% file: wavefunctionansatz.pspdftex
\begin{picture}(0,0)%
\includegraphics{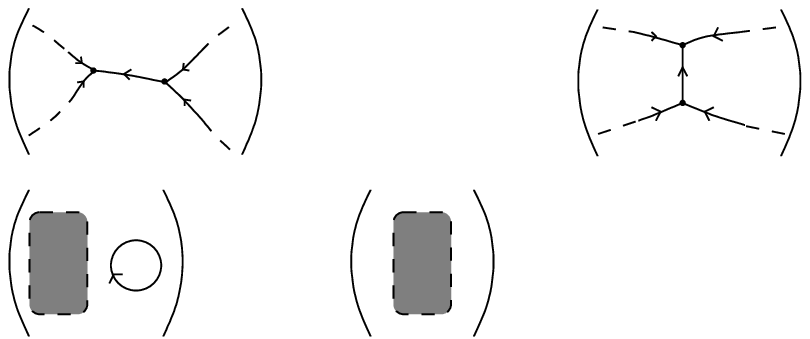}%
\end{picture}%
\setlength{\unitlength}{4144sp}%
\begingroup\makeatletter\ifx\SetFigFont\undefined%
\gdef\SetFigFont#1#2#3#4#5{%
  \reset@font\fontsize{#1}{#2pt}%
  \fontfamily{#3}\fontseries{#4}\fontshape{#5}%
  \selectfont}%
\fi\endgroup%
\begin{picture}(3836,1593)(-104,-991)
\put(973,175){\makebox(0,0)[lb]{\smash{{\SetFigFont{10}{12.0}{\rmdefault}{\mddefault}{\updefault}{\color[rgb]{0,0,0}$l$}%
}}}}
\put(607,-521){\makebox(0,0)[lb]{\smash{{\SetFigFont{10}{12.0}{\rmdefault}{\mddefault}{\updefault}{\color[rgb]{0,0,0}$i$}%
}}}}
\put(644,248){\makebox(0,0)[lb]{\smash{{\SetFigFont{10}{12.0}{\rmdefault}{\mddefault}{\updefault}{\color[rgb]{0,0,0}$m$}%
}}}}
\put(863,-82){\makebox(0,0)[lb]{\smash{{\SetFigFont{10}{12.0}{\rmdefault}{\mddefault}{\updefault}{\color[rgb]{0,0,0}$k$}%
}}}}
\put(387, -8){\makebox(0,0)[lb]{\smash{{\SetFigFont{10}{12.0}{\rmdefault}{\mddefault}{\updefault}{\color[rgb]{0,0,0}$j$}%
}}}}
\put(-89,-741){\makebox(0,0)[lb]{\smash{{\SetFigFont{14}{16.8}{\rmdefault}{\mddefault}{\updefault}{\color[rgb]{0,0,0}$\Phi$}%
}}}}
\put(-89, 65){\makebox(0,0)[lb]{\smash{{\SetFigFont{14}{16.8}{\rmdefault}{\mddefault}{\updefault}{\color[rgb]{0,0,0}$\Phi$}%
}}}}
\put(3244,175){\makebox(0,0)[lb]{\smash{{\SetFigFont{10}{12.0}{\rmdefault}{\mddefault}{\updefault}{\color[rgb]{0,0,0}$n$}%
}}}}
\put(3281,468){\makebox(0,0)[lb]{\smash{{\SetFigFont{10}{12.0}{\rmdefault}{\mddefault}{\updefault}{\color[rgb]{0,0,0}$l$}%
}}}}
\put(3281,-155){\makebox(0,0)[lb]{\smash{{\SetFigFont{10}{12.0}{\rmdefault}{\mddefault}{\updefault}{\color[rgb]{0,0,0}$k$}%
}}}}
\put(361,389){\makebox(0,0)[lb]{\smash{{\SetFigFont{10}{12.0}{\rmdefault}{\mddefault}{\updefault}{\color[rgb]{0,0,0}$i$}%
}}}}
\put(2881,479){\makebox(0,0)[lb]{\smash{{\SetFigFont{10}{12.0}{\rmdefault}{\mddefault}{\updefault}{\color[rgb]{0,0,0}$i$}%
}}}}
\put(2881,-196){\makebox(0,0)[lb]{\smash{{\SetFigFont{10}{12.0}{\rmdefault}{\mddefault}{\updefault}{\color[rgb]{0,0,0}$j$}%
}}}}
\put(1396, 74){\makebox(0,0)[lb]{\smash{{\SetFigFont{14}{16.8}{\rmdefault}{\mddefault}{\updefault}{\color[rgb]{0,0,0}$=\sum_n F^{ijm}_{kln}\,\Phi$}%
}}}}
\put(991,-736){\makebox(0,0)[lb]{\smash{{\SetFigFont{14}{16.8}{\rmdefault}{\mddefault}{\updefault}{\color[rgb]{0,0,0}$=d_i\,\Phi$}%
}}}}
\end{picture}%

%% file: unimodalmove.pspdftex
\begin{picture}(0,0)%
\includegraphics{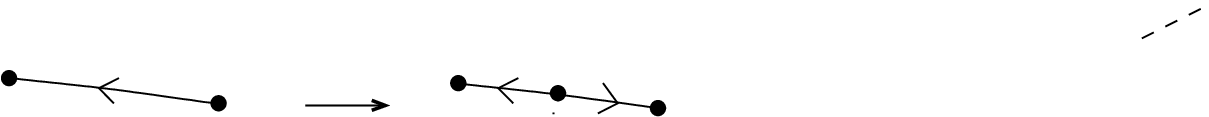}%
\end{picture}%
\setlength{\unitlength}{4144sp}%
\begingroup\makeatletter\ifx\SetFigFont\undefined%
\gdef\SetFigFont#1#2#3#4#5{%
  \reset@font\fontsize{#1}{#2pt}%
  \fontfamily{#3}\fontseries{#4}\fontshape{#5}%
  \selectfont}%
\fi\endgroup%
\begin{picture}(5552,532)(1727,-879)
\put(3961,-646){\makebox(0,0)[lb]{\smash{{\SetFigFont{12}{14.4}{\rmdefault}{\mddefault}{\updefault}{\color[rgb]{0,0,0}$m$}%
}}}}
\put(4456,-691){\makebox(0,0)[lb]{\smash{{\SetFigFont{12}{14.4}{\rmdefault}{\mddefault}{\updefault}{\color[rgb]{0,0,0}$m^*$}%
}}}}
\put(2161,-646){\makebox(0,0)[lb]{\smash{{\SetFigFont{12}{14.4}{\rmdefault}{\mddefault}{\updefault}{\color[rgb]{0,0,0}$m$}%
}}}}
\end{picture}%

%% file: turaevbandbreak.pspdftex
\begin{picture}(0,0)%
\includegraphics{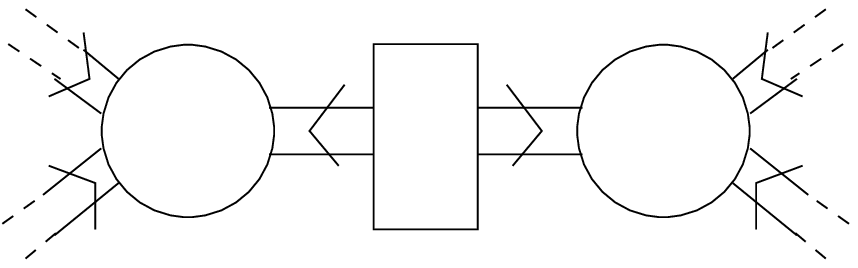}%
\end{picture}%
\setlength{\unitlength}{4144sp}%
\begingroup\makeatletter\ifx\SetFigFont\undefined%
\gdef\SetFigFont#1#2#3#4#5{%
  \reset@font\fontsize{#1}{#2pt}%
  \fontfamily{#3}\fontseries{#4}\fontshape{#5}%
  \selectfont}%
\fi\endgroup%
\begin{picture}(3894,1390)(259,-1214)
\put(1621,-241){\makebox(0,0)[lb]{\smash{{\SetFigFont{12}{14.4}{\rmdefault}{\mddefault}{\updefault}{\color[rgb]{0,0,0}$m$}%
}}}}
\put(2566,-241){\makebox(0,0)[lb]{\smash{{\SetFigFont{12}{14.4}{\rmdefault}{\mddefault}{\updefault}{\color[rgb]{0,0,0}$m^*$}%
}}}}
\put(631, 29){\makebox(0,0)[lb]{\smash{{\SetFigFont{12}{14.4}{\rmdefault}{\mddefault}{\updefault}{\color[rgb]{0,0,0}$i$}%
}}}}
\put(631,-1141){\makebox(0,0)[lb]{\smash{{\SetFigFont{12}{14.4}{\rmdefault}{\mddefault}{\updefault}{\color[rgb]{0,0,0}$j$}%
}}}}
\put(3601,-1096){\makebox(0,0)[lb]{\smash{{\SetFigFont{12}{14.4}{\rmdefault}{\mddefault}{\updefault}{\color[rgb]{0,0,0}$k$}%
}}}}
\put(3691, 29){\makebox(0,0)[lb]{\smash{{\SetFigFont{12}{14.4}{\rmdefault}{\mddefault}{\updefault}{\color[rgb]{0,0,0}$l$}%
}}}}
\put(2161,-196){\rotatebox{270.0}{\makebox(0,0)[lb]{\smash{{\SetFigFont{12}{14.4}{\rmdefault}{\mddefault}{\updefault}{\color[rgb]{0,0,0}$(w_{m^*})^{-1}$}%
}}}}}
\end{picture}%